\documentclass[twoside, journal]{IEEEtran}

\usepackage{epsfig}
\usepackage{graphicx}
\usepackage{amsmath}
\usepackage{amssymb}
\usepackage{float}
\usepackage{url}
\newtheorem{definition}{Definition}
\newtheorem{lemma}{Lemma}
\newtheorem{proposition}{Proposition}

\newtheorem{theorem}{Theorem}

\newtheorem{remark}{Remark}
\newtheorem{example}{Example}

\newcommand{\naturals}{\ensuremath{\mathbb{N}}}
\newcommand{\reals}{\ensuremath{\mathbb{R}}}
\newcommand{\pr}{\ensuremath{\mathbb{P}}}
\newcommand{\expectation}{\ensuremath{\mathbb{E}}}
\newcommand{\LL}{\ensuremath{\mathbb{L}}}

\begin{document}

\title{\huge{On Concentration and Revisited Large Deviations
Analysis of Binary Hypothesis Testing}}

\author{\IEEEauthorblockN{Igal Sason\\
\hspace*{-0.4cm} \url{sason@ee.technion.ac.il}\\
Department of Electrical Engineering\\
Technion -- Israel Institute of Technology\\
Haifa 32000, Israel}}

\maketitle

\begin{abstract}
This paper first introduces a refined version of the
Azuma-Hoeffding inequality for discrete-parameter martingales with
uniformly bounded jumps. The refined inequality is used to revisit
the large deviations analysis of binary hypothesis testing.
\end{abstract}


\section{Introduction}
\label{section: introduction}

An analysis of binary hypothesis testing from an
information-theoretic point of view, and a derivation of its
related error exponents in analogy to optimum channel codes was
provided in \cite{Blahut_IT74}. A nice exposition of the subject
is also provided in \cite[Chapter~11]{Cover and Thomas} where the
exact error exponents for the large deviation analysis of binary
hypothesis testing are provided in terms of relative entropies.

The Azuma-Hoeffding inequality is by now a well-known methodology
that has been often used to prove concentration of measure
phenomena. It is due to Hoeffding \cite{Hoeffding} who proved it
first for a sum of independent and bounded RVs, and Azuma
\cite{Azuma} who later extended it to bounded-difference
martingales. For a nice exposition of the martingale approach,
used for establishing concentration inequalities, the reader is
referred to e.g. \cite{survey2006} and \cite{McDiarmid_tutorial}.
The starting point of this work is an introduction of a known
concentration inequality for discrete-parameter martingales with
uniformly bounded jumps, which forms a refined version of the
Azuma-Hoeffding inequality. It is then used to study some of its
information-theoretic implications in the context of binary
hypothesis testing. Specifically, the tightness of this
concentration inequality is studied via a large deviations
analysis for binary hypothesis testing, and the demonstration of
its improved tightness over the Azuma-Hoeffding inequality is
revisited in this context. Some links of the derived lower bounds
on the error exponents to some information measures (e.g., the
relative entropy and Fisher information) are obtained along the
way.

This paper is structured as follows: Section~\ref{section:
Concentration Inequalities via Martingales} introduces briefly
some preliminary material related to martingales and Azuma's
inequality, and then it considers a refined version of Azuma's
inequality. This refined inequality is followed by a study of some
of its relation to the martingale central limit theorem.
Section~\ref{section: binary hypothesis testing} considers the
relation of the Azuma's inequality and the refined version of this
inequality (which was introduced in Section~\ref{section:
Concentration Inequalities via Martingales}) to large, moderate
and small deviations analysis of binary hypothesis testing.
Section~\ref{section: summary} concludes the paper, followed by
some proofs and complementary details that are relegated to the
appendices.

\section{Preliminaries and a New Concentration Inequality}
\label{section: Concentration Inequalities via Martingales}

In the following, we present briefly essential background on the
martingale approach that is used in this paper to derive
concentration inequalities. A refined version of Azuma's
inequality is then introduced. This concentration inequality is
applied in the next section for revising the large deviations
analysis of binary hypothesis testing.

\subsection{Doob's Martingales} \label{subsection: Martingales}
This sub-section provides a short background on martingales to set
definitions and notation. For a more thorough study of
martingales, the reader it referred to, e.g., \cite{Billingsley}.

\begin{definition}{\bf[Doob's Martingale]} Let $(\Omega, \mathcal{F},
\pr)$ be a probability space. A Doob's martingale sequence is a
sequence $X_0, X_1, \ldots$ of random variables (RVs) and
corresponding sub $\sigma$-algebras $\mathcal{F}_0, \mathcal{F}_1,
\ldots$ (also denoted by $\{X_i, \mathcal{F}_i\}$) that satisfy
the following conditions:
\begin{enumerate}
\item $X_i \in \LL^1(\Omega, \mathcal{F}_i, \pr)$ for every $i$, i.e., each $X_i$
is defined on the same sample space $\Omega$, it is measurable
with respect to the corresponding $\sigma$-algebra $\mathcal{F}_i$
(i.e., $X_i$ is $\mathcal{F}_i$-measurable) and $\expectation
[|X_i|] = \int_{\Omega} |X_i(\omega)| d\pr(\omega) < \infty.$
\item $\mathcal{F}_0 \subseteq \mathcal{F}_1 \subseteq \ldots $ (where this sequence
of $\sigma$-algebras is called a filtration).
\item $X_i = \expectation[ X_{i+1} | \mathcal{F}_i]$
holds almost surely (a.s.) for every $i$.
\end{enumerate}
\label{definition: Doob's martingales}
\end{definition}

For preliminary material on the construction of discrete-time
martingales, see Appendix~\ref{Appendix: remarks on Doob's
martingales} (which is relevant to the analysis in
Section~\ref{section: binary hypothesis testing}).

\subsection{Azuma's Inequality} \label{subsection:
Azuma's inequality} Azuma's inequality\footnote{Azuma's inequality
is also known as the Azuma-Hoeffding inequality. Since this
inequality is referred several times in this paper, it will be
named from this point as Azuma's inequality for the sake of
brevity.} forms a useful concentration inequality for
bounded-difference martingales \cite{Azuma}. In the following,
this inequality is introduced. The reader is referred to, e.g.,
\cite[Chapter~11]{AlonS_tpm3}, \cite{survey2006} and
\cite{McDiarmid_tutorial} for surveys on concentration
inequalities for (sub/ super) martingales.

\begin{theorem}{\bf[Azuma's inequality]}
Let $\{X_k, \mathcal{F}_k\}_{k=0}^{\infty}$ be a
discrete-parameter real-valued martingale sequence such that for
every $k \in \naturals$, the condition $ |X_k - X_{k-1}| \leq d_k$
holds a.s. for some non-negative constants
$\{d_k\}_{k=1}^{\infty}$. Then
\begin{equation}
\pr( | X_n - X_0 | \geq r) \leq 2 \exp\left(-\frac{r^2}{2
\sum_{k=1}^n d_k^2}\right) \, \quad \forall \, r \geq 0.
\label{eq: Azuma's concentration inequality - general case}
\end{equation}
\label{theorem: Azuma's concentration inequality}
\end{theorem}

The concentration inequality stated in Theorem~\ref{theorem:
Azuma's concentration inequality} was proved in \cite{Hoeffding}
for independent bounded random variables, followed by a discussion
on sums of dependent random variables; this inequality was later
derived in \cite{Azuma} for bounded-difference martingales. For a
proof of Theorem~\ref{theorem: Azuma's concentration inequality}
see, e.g., \cite{survey2006} and
\cite[Chapter~2.4]{Dembo_Zeitouni}.

\subsection{A Refined Version of Azuma's Inequality}

\begin{theorem}
Let $\{X_k, \mathcal{F}_k\}_{k=0}^{\infty}$ be a
discrete-parameter real-valued martingale. Assume that, for some
constants $d, \sigma > 0$, the following two requirements are
satisfied a.s.
\begin{eqnarray*}
&& | X_k - X_{k-1} | \leq d, \\
&& \text{Var} (X_k | \mathcal{F}_{k-1}) = \expectation \bigl[(X_k
- X_{k-1})^2 \, | \, \mathcal{F}_{k-1} \bigr] \leq \sigma^2
\end{eqnarray*}
for every $k \in \{1, \ldots, n\}$. Then, for every $\alpha \geq
0$,
\begin{equation}
\hspace*{-0.2cm} \pr(|X_n-X_0| \geq \alpha n) \leq 2 \exp\left(-n
\, D\biggl(\frac{\delta+\gamma}{1+\gamma} \Big|\Big|
\frac{\gamma}{1+\gamma}\biggr) \right) \label{eq: first refined
concentration inequality}
\end{equation}
where
\begin{equation}
\gamma \triangleq \frac{\sigma^2}{d^2}, \quad \delta \triangleq
\frac{\alpha}{d}  \label{eq: notation}
\end{equation}
and
\begin{equation}
D(p || q) \triangleq p \ln\Bigl(\frac{p}{q}\Bigr) + (1-p)
\ln\Bigl(\frac{1-p}{1-q}\Bigr), \quad \forall \, p, q \in [0,1]
\label{eq: divergence}
\end{equation}
is the divergence (a.k.a. relative entropy or Kullback-Leibler
distance) between the two probability distributions $(p,1-p)$ and
$(q,1-q)$. If $\delta>1$, then the probability on the left-hand
side of \eqref{eq: first refined concentration inequality} is
equal to zero. \label{theorem: first refined concentration
inequality}
\end{theorem}
\begin{proof}
The idea of the proof of Theorem~\ref{theorem: first refined
concentration inequality} is essentially similar to the proof of
\cite[Corollary~2.4.7]{Dembo_Zeitouni}. The full proof is provided
in \cite[Section~III]{Sason_submitted_paper}.
\end{proof}

\begin{proposition}
Let $\{X_k, \mathcal{F}_k\}_{k=0}^{\infty}$ be a
discrete-parameter real-valued martingale. Then, for every $\alpha
\geq 0$,
\begin{equation}
\pr(|X_n-X_0| \geq \alpha \sqrt{n}) \leq 2
\exp\Bigl(-\frac{\delta^2}{2\gamma}\Bigr) \Bigl(1+
O\bigl(n^{-\frac{1}{2}}\bigr)\Bigr). \label{eq: concentration1}
\end{equation}
\label{proposition: a similar scaling of the concentration
inequalities}
\end{proposition}
\begin{proof}
This inequality follows from Theorem~\ref{theorem: first refined
concentration inequality} (see
\cite[Appendix~G]{Sason_submitted_paper}).
\end{proof}

\section{Binary Hypothesis Testing}
\label{section: binary hypothesis testing}

Binary hypothesis testing for finite alphabet models was analyzed
via the method of types, e.g., in \cite[Chapter~11]{Cover and
Thomas} and \cite{Csiszar_Shields_FnT}. It is assumed that the
data sequence is of a fixed length $(n)$, and one wishes to make
the optimal decision (based on the Neyman-Pearson ratio test)
based on the received sequence.

Let the RVs $X_1, X_2 ....$ be i.i.d. $\sim Q$, and consider two
hypotheses:
\begin{itemize}
\item $H_1:  Q = P_1$.
\item $H_2:  Q = P_2$.
\end{itemize}
For the simplicity of the analysis, let us assume that the RVs are
discrete, and take their values on a finite alphabet $\mathcal{X}$
where $P_1(x), P_2(x) > 0$ for every $x \in \mathcal{X}$.

In the following, let
\begin{equation*}
L(X_1, \ldots, X_n) \triangleq \ln \frac{P_1^n(X_1, \ldots,
X_n)}{P_2^n(X_1, \ldots, X_n)} = \sum_{i=1}^n \ln
\frac{P_1(X_i)}{P_2(X_i)}
\end{equation*}
designate the log-likelihood ratio. By the strong law of large
numbers (SLLN), if hypothesis $H_1$ is true, then a.s.
\begin{equation}
\lim_{n \rightarrow \infty} \frac{L(X_1, \ldots, X_n)}{n} = D(P_1
|| P_2) \label{eq: a.s. limit of the normalized LLR under
hypothesis H1}
\end{equation}
and otherwise, if hypothesis $H_2$ is true, then a.s.
\begin{equation}
\lim_{n \rightarrow \infty} \frac{L(X_1, \ldots, X_n)}{n} = -D(P_2
|| P_1) \label{eq: a.s. limit of the normalized LLR under
hypothesis H2}
\end{equation}
where the above assumptions on the probability mass functions
$P_1$ and $P_2$ imply that the relative entropies, $D(P_1 || P_2)$
and $D(P_2 || P_1)$, are both finite. Consider the case where for
some fixed constants $\overline{\lambda}, \underline{\lambda} \in
\reals$ where $$-D(P_2||P_1) < \underline{\lambda} \leq
\overline{\lambda} < D(P_1||P_2)$$ one decides on hypothesis $H_1$
if $$ L(X_1, \ldots, X_n) > n \overline{\lambda} $$ and on
hypothesis $H_2$ if $$ L(X_1, \ldots, X_n) < n
\underline{\lambda}.$$ Note that if $\overline{\lambda} =
\underline{\lambda} \triangleq \lambda$ then a decision on the two
hypotheses is based on comparing the normalized log-likelihood
ratio (w.r.t. $n$) to a single threshold $(\lambda)$, and deciding
on hypothesis $H_1$ or $H_2$ if this normalized log-likelihood
ratio is, respectively, above or below $\lambda$. If
$\underline{\lambda} < \overline{\lambda}$ then one decides on
$H_1$ or $H_2$ if the normalized log-likelihood ratio is,
respectively, above the upper threshold $\overline{\lambda}$ or
below the lower threshold $\underline{\lambda}$. Otherwise, if the
normalized log-likelihood ratio is between the upper and lower
thresholds, then an erasure is declared and no decision is taken
in this case.

Let
\begin{eqnarray}
&& \alpha_n^{(1)} \triangleq P_1^n \Bigl( L(X_1, \ldots, X_n) \leq
n \overline{\lambda} \Bigr)
\label{eq: error and erasure event under hypothesis H1} \\
&& \alpha_n^{(2)} \triangleq P_1^n \Bigl( L(X_1, \ldots, X_n) \leq
n \underline{\lambda} \Bigr) \label{eq: error event under
hypothesis H1}
\end{eqnarray}
and
\begin{eqnarray}
&& \beta_n^{(1)}  \triangleq P_2^n \Bigl( L(X_1, \ldots, X_n) \geq
n \underline{\lambda} \Bigr)
\label{eq: error and erasure event under hypothesis H2} \\
&& \beta_n^{(2)}  \triangleq P_2^n \Bigl( L(X_1, \ldots, X_n) \geq
n \overline{\lambda} \Bigr) \label{eq: error event under
hypothesis H2}
\end{eqnarray}
then $\alpha_n^{(1)}$ and $\beta_n^{(1)}$ are the probabilities of
either making an error or declaring an erasure under,
respectively, hypotheses $H_1$ and $H_2$; similarly
$\alpha_n^{(2)}$ and $\beta_n^{(2)}$ are the probabilities of
making an error under hypotheses $H_1$ and $H_2$, respectively.

Let $\pi_1, \pi_2 \in (0,1)$ denote the a-priori probabilities of
the hypotheses $H_1$ and $H_2$, respectively, so
\begin{equation}
P_{\text{e}, n}^{(1)} = \pi_1 \alpha_n^{(1)} + \pi_2 \beta_n^{(1)}
\label{eq: overall probability of a mixed error and erasure event}
\end{equation}
is the probability of having either an error or an erasure, and
\begin{equation}
P_{\text{e}, n}^{(2)} = \pi_1 \alpha_n^{(2)} + \pi_2 \beta_n^{(2)}
\label{eq: overall error probability}
\end{equation}
is the probability of error.

\subsection{Exact Exponents} When we let $n$ tend to infinity,
the exact exponents of $\alpha_n^{(j)}$ and $\beta_n^{(j)}$
($j=1,2$) are derived via Cram\'{e}r's theorem. The resulting
exponents form a straightforward generalization of, e.g.,
\cite[Theorem~3.4.3]{Dembo_Zeitouni} and
\cite[Theorem~6.4]{Hollander_book_2000} that addresses the case
where the decision is made based on a single threshold of the
log-likelihood ratio. In this particular case where
$\overline{\lambda} = \underline{\lambda} \triangleq \lambda$, the
option of erasures does not exist, and $P_{\text{e}, n}^{(1)} =
P_{\text{e}, n}^{(2)} \triangleq P_{\text{e}, n}$ is the error
probability.

In the considered general case with erasures, let $$ \lambda_1
\triangleq -\overline{\lambda}, \quad \lambda_2 \triangleq
-\underline{\lambda} $$ then Cram\'{e}r's theorem on $\reals$
yields that the exact exponents of $\alpha_n^{(1)}$,
$\alpha_n^{(2)}$, $\beta_n^{(1)}$ and $\beta_n^{(2)}$ are given by
\begin{eqnarray}
&& \lim_{n \rightarrow \infty} -\frac{\ln \alpha_n^{(1)}}{n} =
I(\lambda_{1})
\label{eq: exponent of alpha_n1} \\[0.1cm]
&& \lim_{n \rightarrow \infty} -\frac{\ln \alpha_n^{(2)}}{n} =
I(\lambda_{2})
\label{eq: exponent of alpha_n2} \\[0.1cm]
&& \lim_{n \rightarrow \infty} -\frac{\ln \beta_n^{(1)}}{n} =
I(\lambda_2) - \lambda_2 \label{eq: exponent of beta_n1} \\[0.1cm]
&& \lim_{n \rightarrow \infty} -\frac{\ln \beta_n^{(2)}}{n} =
I(\lambda_1) - \lambda_1 \label{eq: exponent of beta_n2}
\end{eqnarray}
where the rate function $I$ is given by
\begin{equation}
I(r) \triangleq \sup_{t \in \reals} \bigl( tr - H(t) \bigr)
\label{eq: rate function}
\end{equation}
and
\begin{equation}
H(t) = \ln \Biggl(\sum_{x \in \mathcal{X}} P_1(x)^{1-t} P_2(x)^t
\Biggr), \quad \forall \, t \in \reals. \label{eq: H}
\end{equation}
The rate function $I$ is convex, lower semi-continuous (l.s.c.)
and non-negative (see, e.g., \cite{Dembo_Zeitouni} and
\cite{Hollander_book_2000}). Note that
$$H(t) = (t-1) D_t(P_2||P_1)$$ where $D_t(P||Q)$ designates R\'{e}yni's
information divergence of order $t$, and $I$ in \eqref{eq: rate
function} is the Fenchel-Legendre transform of $H$ (see, e.g.,
\cite[Definition~2.2.2]{Dembo_Zeitouni}).

From \eqref{eq: overall probability of a mixed error and erasure
event}-- \eqref{eq: exponent of beta_n2}, the exact exponents of
$P_{\text{e}, n}^{(1)}$ and $P_{\text{e}, n}^{(2)}$ are equal to
\begin{equation}
\lim_{n \rightarrow \infty} - \frac{\ln P_{\text{e}, n}^{(1)}}{n}
= \min \Bigl\{ I(\lambda_1), I(\lambda_2) - \lambda_2 \Bigr\}
\label{eq: exact exponent of the overall error and erasure
probability}
\end{equation}
and
\begin{equation}
\lim_{n \rightarrow \infty} - \frac{\ln P_{\text{e}, n}^{(2)}}{n}
= \min \Bigl\{ I(\lambda_2), I(\lambda_1) - \lambda_1 \Bigr\}.
\label{eq: exact error exponent}
\end{equation}

For the case where the decision is based on a single threshold for
the log-likelihood ratio (i.e., $\lambda_1 = \lambda_2 \triangleq
\lambda$), then $P_{\text{e}, n}^{(1)} = P_{\text{e}, n}^{(2)}
\triangleq P_{\text{e}, n}$, and its error exponent is equal to
\begin{equation}
\lim_{n \rightarrow \infty} - \frac{\ln P_{\text{e}, n}}{n} = \min
\Bigl\{ I(\lambda), I(\lambda) - \lambda \Bigr\} \label{eq: exact
exponent of the error probability for a single threshold}
\end{equation}
which coincides with the error exponent in
\cite[Theorem~3.4.3]{Dembo_Zeitouni} (or
\cite[Theorem~6.4]{Hollander_book_2000}). The optimal threshold
for obtaining the best error exponent of the error probability
$P_{\text{e}, n}$ is equal to zero (i.e., $\lambda=0$); in this
case, the exact error exponent is equal to
\begin{equation}
\hspace*{-0.2cm} I(0) = -\min_{0 \leq t \leq 1} \ln \Biggl(
\sum_{x \in \mathcal{X}} P_1(x)^{1-t} P_2(x)^t \Biggr) \triangleq
C(P_1, P_2) \label{eq: Chernoff information}
\end{equation}
which is the Chernoff information of the probability measures
$P_1$ and $P_2$ (see \cite[Eq.~(11.239)]{Cover and Thomas}), and
it is symmetric (i.e., $C(P_1, P_2) = C(P_2, P_1)$). Note that,
from \eqref{eq: rate function}, $I(0) = \sup_{t \in
\reals}\bigl(-H(t)\bigr) = -\inf_{t \in \reals}\bigl(H(t)\bigr)$;
the minimization in \eqref{eq: Chernoff information} over the
interval $[0,1]$ (instead of taking the infimum of $H$ over
$\reals$) is due to the fact that $H(0) = H(1) = 0$ and the
function $H$ in \eqref{eq: H} is convex, so it is enough to
restrict the infimum of $H$ to the closed interval $[0,1]$ for
which it turns to be a minimum.

\subsection{Lower Bound on the Exponents via
Theorem~\ref{theorem: first refined concentration inequality}} In
the following, the tightness of Theorem~\ref{theorem: first
refined concentration inequality} is examined by using it for the
derivation of lower bounds on the error exponent and the exponent
of the event of having either an error or an erasure. These
results will be compared in the next sub-section to the exact
exponents from the previous sub-section.

We first derive a lower bound on the exponent of $\alpha_n^{(1)}$.
Under hypothesis $H_1$, let us construct the martingale sequence
$\{U_k, \mathcal{F}_k\}_{k=0}^n$ where $\mathcal{F}_0 \subseteq
\mathcal{F}_1 \subseteq \ldots \mathcal{F}_n$ is the filtration
$$ \mathcal{F}_0 = \{\emptyset, \Omega\}, \quad \mathcal{F}_k =
\sigma(X_1, \ldots, X_k), \; \; \forall \, k \in \{1, \ldots,
n\}$$ and
\begin{equation}
U_k = \expectation_{P_1^n} \bigl[ L(X_1, \ldots, X_n) \; | \;
\mathcal{F}_k \bigr]. \label{eq: martingale sequence U under
hypothesis H1}
\end{equation}
For every $k \in \{0, \ldots, n\}$
\begin{eqnarray*}
&& U_k = \expectation_{P_1^n} \Biggl[  \sum_{i=1}^n
\ln \frac{P_1(X_i)}{P_2(X_i)} \; \Big| \; \mathcal{F}_k  \Biggr] \\
&& \hspace*{0.5cm} =  \sum_{i=1}^k \ln \frac{P_1(X_i)}{P_2(X_i)} +
\sum_{i=k+1}^n \expectation_{P_1^n} \Biggl[
\ln \frac{P_1(X_i)}{P_2(X_i)} \Biggr] \\
&& \hspace*{0.5cm} =  \sum_{i=1}^k \ln \frac{P_1(X_i)}{P_2(X_i)} +
(n-k) D(P_1 || P_2).
\end{eqnarray*}
In particular
\begin{eqnarray}
&& U_0 = n D(P_1 || P_2), \label{eq: initial value of the
martingale U that is related to the binary hypothesis testing}   \\
&& U_n = \sum_{i=1}^n \ln \frac{P_1(X_i)}{P_2(X_i)} = L(X_1,
\ldots, X_n) \label{eq: final value of the martingale U that is
related to the binary hypothesis testing}
\end{eqnarray}
and, for every $k \in \{1, \ldots, n\}$,
\begin{equation}
U_k - U_{k-1} = \ln \frac{P_1(X_k)}{P_2(X_k)} - D(P_1 || P_2).
\label{eq: jumps of the martingale U that is related to the binary
hypothesis testing}
\end{equation}
Let
\begin{equation}
d_1 \triangleq \max_{x \in \mathcal{X}} \left| \ln
\frac{P_1(x)}{P_2(x)} - D(P_1 || P_2) \right| \label{eq: d1}
\end{equation}
so $d_1 < \infty$ since by assumption the alphabet set
$\mathcal{X}$ is finite, and $P_1(x), P_2(x) > 0$ for every $x \in
\mathcal{X}$. From \eqref{eq: jumps of the martingale U that is
related to the binary hypothesis testing} and \eqref{eq: d1}
$$|U_k - U_{k-1}| \leq d_1$$ holds a.s. for every $k \in \{1, \ldots,
n\}$, and
\begin{eqnarray}
&& \expectation_{P_1^n} \bigl[ (U_k - U_{k-1})^2 \, | \,
\mathcal{F}_{k-1} \bigr] \nonumber \\
&& = \expectation_{P_1} \left[ \left( \ln
\frac{P_1(X_k)}{P_2(X_k)}
- D(P_1 || P_2) \right)^2 \right] \nonumber \\
&& = \sum_{x \in \mathcal{X}} \left\{ P_1(x) \left( \ln
\frac{P_1(x)}{P_2(x)} - D(P_1 || P_2) \right)^2
\right\} \nonumber \\
&& \triangleq \sigma_1^2. \label{eq: sigma1 squared for the jumps
of the martingale U}
\end{eqnarray}

Let
\begin{eqnarray}
&& \hspace*{-1cm} \varepsilon_{1,1} = D(P_1 || P_2) -
\overline{\lambda}, \quad \varepsilon_{2,1} = D(P_2 || P_1) +
\underline{\lambda}
\label{eq: the epsilons introduced for mixed errors and erasures in binary hypothesis testing} \\
&& \hspace*{-1cm} \varepsilon_{1,2} = D(P_1 || P_2) -
\underline{\lambda}, \quad \varepsilon_{2,2} = D(P_2 || P_1) +
\overline{\lambda} \label{eq: the epsilons introduced for errors
in binary hypothesis testing}
\end{eqnarray}
The probability of making an erroneous decision on hypothesis
$H_2$ or declaring an erasure under the hypothesis $H_1$ is equal
to $\alpha_n^{(1)}$, and from Theorem~\ref{theorem: first refined
concentration inequality}
\begin{eqnarray}
&& \alpha_n^{(1)} \triangleq P_1^n \bigl( L(X_1, \ldots, X_n)
\leq n \overline{\lambda} \bigr) \nonumber \\
&& \hspace*{0.7cm} \stackrel{\text{(a)}}{=} P_1^n(U_n - U_0 \leq
-\varepsilon_{1,1} \, n)
\label{eq: intermediate step in the derivation of a bound on alpha1} \\
&& \hspace*{0.7cm} \stackrel{\text{(b)}}{\leq} \exp \left(-n \,
D\Bigl(\frac{\delta_{1,1} + \gamma_1}{1+\gamma_1} \Big|\Big|
\frac{\gamma_1}{1+\gamma_1} \Bigr) \right) \label{eq:
concentration inequality for the first error event}
\end{eqnarray}
where equality~(a) follows from \eqref{eq: initial value of the
martingale U that is related to the binary hypothesis testing},
\eqref{eq: final value of the martingale U that is related to the
binary hypothesis testing} and \eqref{eq: the epsilons introduced
for mixed errors and erasures in binary hypothesis testing}, and
inequality~(b) follows from Theorem~\ref{theorem: first refined
concentration inequality} with
\begin{equation}
\gamma_1 \triangleq \frac{\sigma_1^2}{d_1^2}, \quad \delta_{1,1}
\triangleq \frac{\varepsilon_{1,1}}{d_1}. \label{eq: gamma1 and
delta1,1}
\end{equation}
Note that if $\varepsilon_{1,1} > d_1$ then it follows from
\eqref{eq: jumps of the martingale U that is related to the binary
hypothesis testing} and \eqref{eq: d1} that $\alpha_n^{(1)}$ is
zero; in this case $\delta_{1,1} > 1$, so the divergence in
\eqref{eq: concentration inequality for the first error event} is
infinity and the upper bound is also equal to zero. Hence, it is
assumed without loss of generality that $\delta_{1,1} \in [0, 1]$.

Similarly to \eqref{eq: martingale sequence U under hypothesis
H1}, under hypothesis~$H_2$, let us define the martingale sequence
$\{U_k, \mathcal{F}_k\}_{k=0}^n$ with the same filtration and
\begin{equation}
\hspace*{-0.4cm} U_k = \expectation_{P_2^n} \bigl[ L(X_1, \ldots,
X_n) \; | \; \mathcal{F}_k \bigr], \quad \forall \, k \in \{0,
\ldots, n\}. \label{eq: martingale sequence U  under hypothesis
H2}
\end{equation}
For every $k \in \{0, \ldots, n\}$
\begin{eqnarray*}
&& U_k = \sum_{i=1}^k \ln \frac{P_1(X_i)}{P_2(X_i)} - (n-k) D(P_2
|| P_1)
\end{eqnarray*}
and in particular
\begin{equation}
U_0 = -n D(P_2 || P_1), \quad U_n = L(X_1, \ldots, X_n).
\label{eq: initial and final values of the second martingale
sequence}
\end{equation}
For every $k \in \{1, \ldots, n\}$,
\begin{equation}
U_k - U_{k-1} = \ln \frac{P_1(X_k)}{P_2(X_k)} + D(P_2 || P_1).
\label{eq: jumps of the martingale U under hypothesis H2}
\end{equation}
Let
\begin{equation}
d_2 \triangleq \max_{x \in \mathcal{X}} \left| \ln
\frac{P_2(x)}{P_1(x)} - D(P_2 || P_1) \right| \label{eq: d2}
\end{equation}
then, the jumps of the latter martingale sequence are uniformly
bounded by $d_2$ and, similarly to \eqref{eq: sigma1 squared for
the jumps of the martingale U}, for every $k \in \{1, \ldots, n\}$
\begin{eqnarray}
&& \expectation_{P_2^n} \bigl[ (U_k - U_{k-1})^2 \, | \,
\mathcal{F}_{k-1} \bigr] \nonumber \\
&& = \sum_{x \in \mathcal{X}} \left\{ P_2(x) \left( \ln
\frac{P_2(x)}{P_1(x)} - D(P_2 || P_1) \right)^2 \right\} \nonumber
\\ && \triangleq \sigma_2^2. \label{eq: sigma2 squared for the
jumps of the martingale U}
\end{eqnarray}
Hence, it follows from Theorem~\ref{theorem: first refined
concentration inequality} that
\begin{eqnarray}
&& \beta_n^{(1)} \triangleq P_2^n \bigl( L(X_1, \ldots, X_n)
\geq n \underline{\lambda} \bigr) \nonumber \\
&& \hspace*{0.7cm} = P_2^n(U_n - U_0 \geq \varepsilon_{2,1} \, n)
\label{eq: intermediate step in the derivation of a bound on beta1} \\
&& \hspace*{0.7cm} \leq  \exp \left(-n \,
D\Bigl(\frac{\delta_{2,1} + \gamma_2}{1+\gamma_2} \Big|\Big|
\frac{\gamma_2}{1+\gamma_2} \Bigr) \right) \label{eq:
concentration inequality for the second error event}
\end{eqnarray}
where the equality in \eqref{eq: intermediate step in the
derivation of a bound on beta1} holds due to \eqref{eq: initial
and final values of the second martingale sequence} and \eqref{eq:
the epsilons introduced for mixed errors and erasures in binary
hypothesis testing}, and \eqref{eq: concentration inequality for
the second error event} follows from Theorem~\ref{theorem: first
refined concentration inequality} with
\begin{equation}
\gamma_2 \triangleq \frac{\sigma_2^2}{d_2^2}, \quad \delta_{2,1}
\triangleq \frac{\varepsilon_{2,1}}{d_2} \label{eq: gamma2 and
delta2,1}
\end{equation}
and $d_2$, $\sigma_2$ are introduced, respectively, in \eqref{eq:
d2} and \eqref{eq: sigma2 squared for the jumps of the martingale
U}.

From \eqref{eq: overall probability of a mixed error and erasure
event}, \eqref{eq: concentration inequality for the first error
event} and \eqref{eq: concentration inequality for the second
error event}, the exponent of the probability of either having an
error or an erasure is lower bounded by
\begin{equation}
\lim_{n \rightarrow \infty} - \frac{\ln P_{\text{e}, n}^{(1)}}{n}
\geq \min_{i=1,2} D\Bigl(\frac{\delta_{i,1} +
\gamma_i}{1+\gamma_i} \Big|\Big| \frac{\gamma_i}{1+\gamma_i}
\Bigr). \label{eq: lower bound on the exponent of mixed errors and
erasures for binary hypothesis testing}
\end{equation}
Similarly to the above analysis, one gets from \eqref{eq: overall
error probability} and \eqref{eq: the epsilons introduced for
errors in binary hypothesis testing} that the error exponent is
lower bounded by
\begin{equation}
\lim_{n \rightarrow \infty} - \frac{\ln P_{\text{e}, n}^{(2)}}{n}
\geq \min_{i=1,2} D\Bigl(\frac{\delta_{i,2} +
\gamma_i}{1+\gamma_i} \Big|\Big| \frac{\gamma_i}{1+\gamma_i}
\Bigr) \label{eq: lower bound on the error exponent for binary
hypothesis testing}
\end{equation}
where
\begin{equation}
\delta_{1,2} \triangleq \frac{\varepsilon_{1,2}}{d_1}, \quad
\delta_{2,2} \triangleq \frac{\varepsilon_{2,2}}{d_2}. \label{eq:
delta1,2 and delta2,2}
\end{equation}

For the case of a single threshold (i.e., $\overline{\lambda} =
\underline{\lambda} \triangleq \lambda$) then \eqref{eq: lower
bound on the exponent of mixed errors and erasures for binary
hypothesis testing} and \eqref{eq: lower bound on the error
exponent for binary hypothesis testing} coincide, and one obtains
that the error exponent satisfies
\begin{equation}
\lim_{n \rightarrow \infty} - \frac{\ln P_{\text{e}, n}}{n} \geq
\min_{i=1,2} D\Bigl(\frac{\delta_i + \gamma_i}{1+\gamma_i}
\Big|\Big| \frac{\gamma_i}{1+\gamma_i} \Bigr) \label{eq: lower
bound on the error exponent for binary hypothesis testing with a
single threshold}
\end{equation}
where $\delta_i$ is the common value of $\delta_{i,1}$ and
$\delta_{i,2}$ (for $i=1,2$). In this special case, the zero
threshold is optimal (see, e.g., \cite[p.~93]{Dembo_Zeitouni}),
which then yields that \eqref{eq: lower bound on the error
exponent for binary hypothesis testing with a single threshold} is
satisfied with
\begin{equation}
\delta_1 = \frac{D(P_1 || P_2)}{d_1}, \quad \delta_2 = \frac{D(P_2
|| P_1)}{d_2} \label{eq: delta1 and delta2}
\end{equation}
with $d_1$ and $d_2$ from \eqref{eq: d1} and \eqref{eq: d2},
respectively. The right-hand side of \eqref{eq: lower bound on the
error exponent for binary hypothesis testing with a single
threshold} forms a lower bound on Chernoff information which is
the exact error exponent for this special case.

\subsection{Comparison of the Lower Bounds on the Exponents with
those that Follow from Azuma's Inequality} The lower bounds on the
error exponent and the exponent of the probability of having
either errors or erasures, that were derived in the previous
sub-section via Theorem~\ref{theorem: first refined concentration
inequality}, are compared in the following to the loosened lower
bounds on these exponents that follow from Azuma's inequality.

We first obtain upper bounds on $\alpha_n^{(1)}, \alpha_n^{(2)},
\beta_n^{(1)}$ and $\beta_n^{(2)}$ via Azuma's inequality, and
then use them to derive lower bounds on the exponents of
$P_{\text{e},n}^{(1)}$ and $P_{\text{e},n}^{(2)}$.

From \eqref{eq: jumps of the martingale U that is related to the
binary hypothesis testing}, \eqref{eq: d1}, \eqref{eq:
intermediate step in the derivation of a bound on alpha1},
\eqref{eq: gamma1 and delta1,1}, and Azuma's inequality
\begin{equation}
\alpha_n^{(1)} \leq \exp \biggl(-\frac{\delta_{1,1}^2 n}{2}
\biggr) \label{eq: Azuma's inequality for alpha1}
\end{equation}
and, similarly, from \eqref{eq: jumps of the martingale U under
hypothesis H2}, \eqref{eq: d2}, \eqref{eq: intermediate step in
the derivation of a bound on beta1}, \eqref{eq: gamma2 and
delta2,1}, and Azuma's inequality
\begin{equation}
\beta_n^{(1)} \leq \exp \biggl(-\frac{\delta_{2,1}^2 n}{2}
\biggr). \label{eq: Azuma's inequality for beta1}
\end{equation}
From \eqref{eq: error event under hypothesis H1}, \eqref{eq: error
event under hypothesis H2}, \eqref{eq: the epsilons introduced for
errors in binary hypothesis testing}, \eqref{eq: delta1,2 and
delta2,2} and Azuma's inequality
\begin{eqnarray}
&& \hspace*{-0.5cm} \alpha_n^{(2)} \leq \exp
\biggl(-\frac{\delta_{1,2}^2 n}{2} \biggr)
\label{eq: Azuma's inequality for alpha2} \\
&& \hspace*{-0.5cm} \beta_n^{(2)} \leq \exp
\biggl(-\frac{\delta_{2,2}^2 n}{2} \biggr). \label{eq: Azuma's
inequality for beta2}
\end{eqnarray}
Therefore, it follows from \eqref{eq: overall probability of a
mixed error and erasure event}, \eqref{eq: overall error
probability} and \eqref{eq: Azuma's inequality for
alpha1}--\eqref{eq: Azuma's inequality for beta2} that the
resulting lower bounds on the exponents of $P_{\text{e},n}^{(1)}$
and $P_{\text{e},n}^{(2)}$ are
\begin{equation}
\lim_{n \rightarrow \infty} - \frac{\ln P_{\text{e}, n}^{(j)}}{n}
\geq \min_{i=1,2} \frac{\delta_{i,j}^2}{2}, \quad j = 1, 2
\label{eq: lower bounds on the exponents for binary hypothesis
testing via Azuma inequality}
\end{equation}
as compared to \eqref{eq: lower bound on the exponent of mixed
errors and erasures for binary hypothesis testing} and \eqref{eq:
lower bound on the error exponent for binary hypothesis testing}
which give, for $j=1,2$,
\begin{equation}
\lim_{n \rightarrow \infty} - \frac{\ln P_{\text{e}, n}^{(j)}}{n}
\geq \min_{i=1,2} D\Bigl(\frac{\delta_{i,j} +
\gamma_i}{1+\gamma_i} \Big|\Big| \frac{\gamma_i}{1+\gamma_i}
\Bigr). \label{eq: lower bounds on the exponents for binary
hypothesis testing via Theorem 2}
\end{equation}
For the specific case of a zero threshold, the lower bound on the
error exponent which follows from Azuma's inequality is given by
\begin{equation}
\lim_{n \rightarrow \infty} - \frac{\ln P_{\text{e}, n}^{(j)}}{n}
\geq \min_{i=1,2} \frac{\delta_i^2}{2} \label{eq: loosened lower
bound on the error exponent for zero threshold}
\end{equation}
with the values of $\delta_1$ and $\delta_2$ in \eqref{eq: delta1
and delta2}.

The lower bounds on the exponents in \eqref{eq: lower bounds on
the exponents for binary hypothesis testing via Azuma inequality}
and \eqref{eq: lower bounds on the exponents for binary hypothesis
testing via Theorem 2} are compared in the following. Note that
the lower bounds in \eqref{eq: lower bounds on the exponents for
binary hypothesis testing via Azuma inequality} are loosened as
compared to those in \eqref{eq: lower bounds on the exponents for
binary hypothesis testing via Theorem 2} since they follow,
respectively, from Azuma's inequality and its improvement in
Theorem~\ref{theorem: first refined concentration inequality}.

The divergence in the exponent of \eqref{eq: lower bounds on the
exponents for binary hypothesis testing via Theorem 2} is equal to

\small \vspace*{-0.2cm}
\begin{eqnarray}
&& \hspace*{-0.8cm} D\Bigl(\frac{\delta_{i,j} +
\gamma_i}{1+\gamma_i}
\Big|\Big| \frac{\gamma_i}{1+\gamma_i} \Bigr) \nonumber \\[0.1cm]
&& \hspace*{-0.8cm} =
\left(\frac{\delta_{i,j}+\gamma_i}{1+\gamma_i} \right) \ln \left(
1 + \frac{\delta_{i,j}}{\gamma_i} \right) +
\left(\frac{1-\delta_{i,j}}{1+\gamma_i}\right)
\ln(1-\delta_{i,j}) \nonumber \\[0.1cm]
&& \hspace*{-0.8cm} = \frac{\gamma_i}{1+\gamma_i} \left[ \left( 1
+ \frac{\delta_{i,j}}{\gamma_i} \right) \ln \Bigl( 1 +
\frac{\delta_{i,j}}{\gamma_i}
\Bigr) + \frac{(1-\delta_{i,j}) \ln(1-\delta_{i,j})}{\gamma_i}\right]. \nonumber \\
\label{eq: equality for the divergence}
\end{eqnarray}

\normalsize
\begin{lemma}
\begin{equation}
(1+u) \ln(1+u) \geq \left\{
\begin{array}{ll}
u + \frac{u^2}{2}, \quad & u \in [-1, 0] \\[0.2cm]
u+\frac{u^2}{2}-\frac{u^3}{6}, \quad & u \geq 0
\end{array}
\right. \label{eq: inequality for lower bounding the divergence}
\end{equation}
where at $u=-1$, the left-hand side is defined to be zero (it is
the limit of this function when $u \rightarrow -1$ from above).
\label{lemma: inequality for lower bounding the divergence}
\end{lemma}
\begin{proof}
The proof follows by elementary calculus. 
\end{proof}

Since $\delta_{i,j} \in [0,1]$, then \eqref{eq: equality for the
divergence} and Lemma~\ref{lemma: inequality for lower bounding
the divergence} imply that
\begin{equation}
D\Bigl(\frac{\delta_{i,j} + \gamma_i}{1+\gamma_i} \Big|\Big|
\frac{\gamma_i}{1+\gamma_i} \Bigr) \geq \frac{\delta_{i,j}^2}{2
\gamma_i} - \frac{\delta_{i,j}^3}{6 \gamma_i^2 (1+\gamma_i)}.
\label{eq: lower bound on the lower bound of the exponents}
\end{equation}
Hence, by comparing \eqref{eq: lower bounds on the exponents for
binary hypothesis testing via Azuma inequality} with the
combination of \eqref{eq: lower bounds on the exponents for binary
hypothesis testing via Theorem 2} and \eqref{eq: lower bound on
the lower bound of the exponents}, then it follows that (up to a
second-order approximation) the lower bounds on the exponents that
were derived via Theorem~\ref{theorem: first refined concentration
inequality} are improved by at least a factor of $\bigl(\max
\gamma_i\bigr)^{-1}$ as compared to those that follow from Azuma's
inequality.

\begin{example}
Consider two probability measures $P_1$ and $P_2$ where
$$P_1(0) = P_2(1) = 0.4, \quad P_1(1) = P_2(0) = 0.6,$$
and the case of a single threshold of the log-likelihood ratio
that is set to zero (i.e., $\lambda = 0$). The exact error
exponent in this case is Chernoff information that is equal to
$$C(P_1, P_2) = 2.04 \cdot 10^{-2}.$$ The improved lower bound on
the error exponent in \eqref{eq: lower bound on the error exponent
for binary hypothesis testing with a single threshold} and
\eqref{eq: delta1 and delta2} is equal to $1.77 \cdot 10^{-2}$,
whereas the loosened lower bound in \eqref{eq: loosened lower
bound on the error exponent for zero threshold} is equal to $1.39
\cdot 10^{-2}$. In this case $\gamma_1 = \frac{2}{3}$ and
$\gamma_2 = \frac{7}{9}$, so the improvement in the lower bound on
the error exponent is indeed by a factor of approximately
$\left(\max_i \gamma_i \right)^{-1} = \frac{9}{7}.$ Note that,
from \eqref{eq: concentration inequality for the first error
event}, \eqref{eq: concentration inequality for the second error
event} and \eqref{eq: Azuma's inequality for alpha1}--\eqref{eq:
Azuma's inequality for beta2}, these are lower bounds on the error
exponents for any finite block length $n$, and not only
asymptotically in the limit where $n \rightarrow \infty$. The
operational meaning of this example is that the improved lower
bound on the error exponent assures that a fixed error probability
can be obtained based on a sequence of i.i.d. RVs whose length is
reduced by 22.2\% as compared to the loosened bound which follows
from Azuma's inequality.
\end{example}

\subsection{Comparison of the Exact and Lower Bounds on the Error
Exponents, Followed by a Relation to Fisher Information} In the
following, we compare the exact and lower bounds on the error
exponents. Consider the case where there is a single threshold on
the log-likelihood ratio (i.e., referring to the case where the
erasure option is not provided) that is set to zero. The exact
error exponent in this case is given by the Chernoff information
(see \eqref{eq: Chernoff information}), and it will be compared to
the two lower bounds on the error exponents that were derived in
the previous two subsections.

Let $\{P_{\theta}\}_{\theta \in \Theta}$, denote an indexed family
of probability mass functions where $\Theta$ denotes the parameter
set. Assume that $P_{\theta}$ is differentiable in the parameter
$\theta$. Then, the Fisher information is defined as
\begin{equation}
J(\theta) \triangleq \expectation_{\theta} \left[
\frac{\partial}{\partial \theta} \, \ln P_{\theta}(x) \right]^2
\label{eq: Fisher information}
\end{equation}
where the expectation is w.r.t. the probability mass function
$P_{\theta}$. The divergence and Fisher information are two
related information measures, satisfying the equality
\begin{equation}
\lim_{\theta' \rightarrow \theta} \frac{D(P_{\theta} ||
P_{\theta'})}{(\theta - \theta')^2} = \frac{J(\theta)}{2}
\label{eq: relation between the divergence and Fisher information}
\end{equation}
(note that if it was a relative entropy to base~2 then the
right-hand side of \eqref{eq: relation between the divergence and
Fisher information} would have been divided by $\ln 2$, and be
equal to $\frac{J(\theta)}{\ln 4}$ as in \cite[Eq.~(12.364)]{Cover
and Thomas}).
\begin{proposition}
Under the above assumptions,
\begin{itemize}
\item The Chernoff information and Fisher information are related
information measures that satisfy the equality
\begin{equation}
\lim_{\theta' \rightarrow \theta} \frac{C(P_{\theta},
P_{\theta'})}{(\theta - \theta')^2} = \frac{J(\theta)}{8}.
\label{eq: relation between the Chernoff information and Fisher
information}
\end{equation}
\item Let
\begin{equation}
E_{\text{L}}(P_{\theta}, P_{\theta'}) \triangleq \min_{i=1,2}
D\Bigl(\frac{\delta_i + \gamma_i}{1+\gamma_i} \Big|\Big|
\frac{\gamma_i}{1+\gamma_i} \Bigr) \label{eq: E_L}
\end{equation}
be the lower bound on the error exponent in \eqref{eq: lower bound
on the error exponent for binary hypothesis testing with a single
threshold} which corresponds to $P_1 \triangleq P_{\theta}$ and
$P_2 \triangleq P_{\theta'}$, then also
\begin{equation}
\lim_{\theta' \rightarrow \theta} \frac{E_{\text{L}}(P_{\theta},
P_{\theta'})}{(\theta - \theta')^2} = \frac{J(\theta)}{8}.
\label{eq: relation between the improved lower bound on the error
exponent and Fisher information}
\end{equation}
\item Let
\begin{equation}
\widetilde{E}_{\text{L}}(P_{\theta}, P_{\theta'}) \triangleq
\min_{i=1,2} \frac{\delta_i^2}{2} \label{eq: tilde E_L}
\end{equation}
be the loosened lower bound on the error exponent in \eqref{eq:
loosened lower bound on the error exponent for zero threshold}
which refers to $P_1 \triangleq P_{\theta}$ and $P_2 \triangleq
P_{\theta'}$. Then,
\begin{equation}
\lim_{\theta' \rightarrow \theta}
\frac{\widetilde{E}_{\text{L}}(P_{\theta}, P_{\theta'})}{(\theta -
\theta')^2} = \frac{a(\theta) \, J(\theta)}{8} \label{eq: relation
between the loosened lower bound on the error exponent and Fisher
information}
\end{equation}
for some deterministic function $a$ bounded in $[0, 1]$, and there
exists an indexed family of probability mass functions for which
$a(\theta)$ can be made arbitrarily close to zero for any fixed
value of $\theta \in \Theta$.
\end{itemize}
\label{proposition: Fisher information}
\end{proposition}
\begin{proof}
See Appendix~\ref{appendix: Fisher information}.
\end{proof}

\vspace*{0.1cm} Proposition~\ref{proposition: Fisher information}
shows that, in the considered setting, the refined lower bound on
the error exponent provides the correct behavior of the error
exponent for a binary hypothesis testing when the relative entropy
between the pair of probability mass functions that characterize
the two hypotheses tends to zero. This stays in contrast to the
loosened error exponent, which follows from Azuma's inequality,
whose scaling may differ significantly from the correct exponent
(for a concrete example, see the last part of the proof in
Appendix~\ref{appendix: Fisher information}).

\begin{example}
Consider the index family of of probability mass functions defined
over the binary alphabet $\mathcal{X} = \{0,1\}$:
$$ P_\theta(0) = 1-\theta, \; \; P_\theta(1) = \theta, \quad
\forall \, \theta \in (0,1).$$ From \eqref{eq: Fisher
information}, the Fisher information is equal to
$$ J(\theta) = \frac{1}{\theta} + \frac{1}{1-\theta}$$
and, at the point $\theta = 0.5$, $J(\theta) = 4$. Let $\theta_1 =
0.51$ and $\theta_2 = 0.49$, so from \eqref{eq: relation between
the Chernoff information and Fisher information} and \eqref{eq:
relation between the improved lower bound on the error exponent
and Fisher information}
$$ C(P_{\theta_1}, P_{\theta_2}), E_{\text{L}}(P_{\theta_1}, P_{\theta_2}) \approx
\frac{J(\theta) (\theta_1 - \theta_2)^2}{8} = 2.00 \cdot
10^{-4}.$$ Indeed, the exact values of $C(P_{\theta_1},
P_{\theta_2})$ and $E_{\text{L}}(P_{\theta_1}, P_{\theta_2})$ are
$2.000 \cdot 10^{-4}$ and $1.997 \cdot 10^{-4}$, respectively.
\end{example}

\section{Summary}
\label{section: summary} This work introduces a concentration
inequality for discrete-parameter martingales with uniformly
bounded jumps, which forms a refined version of Azuma's
inequality. The tightness of this concentration inequality is
studied via a large deviations analysis of binary hypothesis
testing, and the demonstration of its improved tightness over
Azuma's inequality is revisited in this context. Some links of the
derived lower bounds on the error exponents to some information
measures (e.g., the relative entropy and Fisher information) are
obtained along the way. This paper presents in part the work in
\cite{Sason_submitted_paper} where further concentration
inequalities that form a refinement of Azuma's inequality were
derived, followed by some further applications of these
concentration inequalities in information theory, communication,
and coding theory. It is meant to stimulate the use of some
refined versions of the Azuma-Hoeffding inequality in
information-theoretic aspects.

\appendices

\section{Some complementary remarks concerning the
construction of Doob's martingales} \label{Appendix: remarks on
Doob's martingales}

This appendix is relevant to the analysis in Section~\ref{section:
binary hypothesis testing}.

\begin{remark} Let $\{X_i, \mathcal{F}_i\}$ be a martingale sequence.

For every $i$, $\expectation[X_{i+1}] = \expectation \bigl[
\expectation[X_{i+1} | \mathcal{F}_i] \bigr] = \expectation[X_i]$,
so the expectation of a martingale stays constant.
\end{remark}

\begin{remark} One can generate martingale sequences by the
following procedure: Given a RV $X \in \LL^1(\Omega, \mathcal{F},
\pr)$ and an arbitrary filtration of sub $\sigma$-algebras
$\{\mathcal{F}_i\}$, let
\begin{equation*}
X_i = \expectation[X | \mathcal{F}_i] \quad i = 0, 1, \ldots .
\end{equation*}
Then, the sequence $X_0, X_1, \ldots$ forms a martingale since
\begin{enumerate}
\item The RV $X_i = \expectation[X |
\mathcal{F}_i]$ is $\mathcal{F}_i$-measurable, and also
$\expectation[|X_i|] \leq \expectation[|X|] < \infty$ (since
conditioning reduces the expectation of the absolute value).
\item By construction $\{\mathcal{F}_i\}$ is a filtration.
\item From the tower principle for conditional expectations,
since $\{\mathcal{F}_i\}$ is a filtration, then for every $i$
\begin{equation*}
\expectation[X_{i+1} | \mathcal{F}_i] = \expectation \bigl[
\expectation[ X | \mathcal{F}_{i+1}] | \mathcal{F}_i \bigr] =
\expectation[X | \mathcal{F}_i]  \quad \text{a.s}.
\end{equation*}
\end{enumerate}
\label{remark: construction of martingales}
\end{remark}

\begin{remark}
In continuation to Remark~\ref{remark: construction of
martingales}, one can choose $\mathcal{F}_0 = \{\Omega, \emptyset
\}$ and $\mathcal{F}_n = \mathcal{F}$. Hence, $X_0, X_1, \ldots,
X_n$ is a martingale sequence where
\begin{eqnarray*}
&& \hspace*{-0.8cm} X_0 = \expectation[X | \mathcal{F}_0] =
\expectation[X] \quad
\text{(since $X$ is independent of $\mathcal{F}_0$)} \\
&& \hspace*{-0.8cm} X_n = \expectation[X | \mathcal{F}_n] = X \;
\; \text{a.s.} \quad \text{(since $X$ is
$\mathcal{F}$-measurable)}.
\end{eqnarray*}
This has the following interpretation: At the beginning, we don't
know anything about $X$, so it is initially estimated by its
expectation. We then reveal at each step more and more information
about $X$ until we can specify it completely (a.s.).
\label{remark: construction of martingales (cont.)}
\end{remark}

\section{Proof of Proposition~\ref{proposition: Fisher information}}
\label{appendix: Fisher information} The proof of \eqref{eq:
relation between the Chernoff information and Fisher information}
is based on calculus, and it is similar to the proof of the limit
in \eqref{eq: relation between the divergence and Fisher
information} that relates the divergence and Fisher information.
For the proof of \eqref{eq: relation between the improved lower
bound on the error exponent and Fisher information}, note that

\vspace*{-0.3cm} \small
\begin{equation}
\hspace*{-0.25cm} C(P_{\theta}, P_{\theta'}) \geq
E_{\text{L}}(P_{\theta}, P_{\theta'}) \geq \min_{i=1,2}
\left\{\frac{\delta_i^2}{2 \gamma_i} - \frac{\delta_i^3}{6
\gamma_i^2 (1+\gamma_i)} \right\}. \label{eq: bounds on the
improved lower bound}
\end{equation}
\normalsize The left-hand side of \eqref{eq: bounds on the
improved lower bound} holds since $E_{\text{L}}$ is a lower bound
on the error exponent, and the exact value of this error exponent
is the Chernoff information. The right-hand side of \eqref{eq:
bounds on the improved lower bound} follows from Lemma~\ref{lemma:
inequality for lower bounding the divergence} (see \eqref{eq:
lower bound on the lower bound of the exponents}) and the
definition of $E_{\text{L}}$ in \eqref{eq: E_L}. By definition
$\gamma_i \triangleq \frac{\sigma_i^2}{d_i^2}$ and $\delta_i
\triangleq \frac{\varepsilon_i}{d_i}$ where, based on \eqref{eq:
delta1 and delta2},
\begin{equation}
\varepsilon_1 \triangleq D(P_\theta || P_{\theta'}), \quad
\varepsilon_2 \triangleq D(P_\theta' || P_\theta). \label{eq:
epsilon1,2}
\end{equation}
The term on the left-hand side of \eqref{eq: bounds on the
improved lower bound} therefore satisfies
\begin{eqnarray*}
&& \frac{\delta_i^2}{2 \gamma_i} -
\frac{\delta_i^3}{6 \gamma_i^2 (1+\gamma_i)} \\
&& = \frac{\varepsilon_i^2}{2 \sigma_i^2} - \frac{\varepsilon_i^3
d_i^3}{6 \sigma_i^2 (\sigma_i^2 + d_i^2)} \geq
\frac{\varepsilon_i^2}{2 \sigma_i^2} \left( 1 -
\frac{\varepsilon_i d_i}{3} \right)
\end{eqnarray*}
so it follows from \eqref{eq: bounds on the improved lower bound}
and the last inequality that
\begin{equation}
\hspace*{-0.25cm} C(P_{\theta}, P_{\theta'}) \geq
E_{\text{L}}(P_\theta, P_{\theta'}) \geq \min_{i=1,2} \left\{
\frac{\varepsilon_i^2}{2 \sigma_i^2} \left( 1 -
\frac{\varepsilon_i d_i}{3} \right) \right\}. \label{eq: 2nd ver.
for the bounds on the improved lower bound}
\end{equation}
Based on the continuity assumption of the indexed family
$\{P_{\theta}\}_{\theta \in \Theta}$, then it follows from
\eqref{eq: epsilon1,2} that
$$\lim_{\theta' \rightarrow \theta} \varepsilon_i=0,
\quad \forall \, i \in \{ 1, 2\}$$ and also, from \eqref{eq: d1}
and \eqref{eq: d2} with $P_1$ and $P_2$ replaced by $P_\theta$ and
$P_\theta'$ respectively, then
$$\lim_{\theta' \rightarrow \theta} d_i = 0, \quad \forall \, i \in \{ 1, 2\}.$$
It therefore follows from \eqref{eq: relation between the Chernoff
information and Fisher information} and \eqref{eq: 2nd ver. for
the bounds on the improved lower bound} that
\begin{equation}
\hspace*{-0.4cm} \frac{J(\theta)}{8} \geq \lim_{\theta'
\rightarrow \theta} \frac{E_{\text{L}}(P_{\theta},
P_{\theta'})}{(\theta - \theta')^2} \geq \lim_{\theta' \rightarrow
\theta} \, \min_{i=1,2} \left\{ \frac{\varepsilon_i^2}{2
\sigma_i^2 (\theta-\theta')^2} \right\}. \label{eq: bounds on the
limit}
\end{equation}
The idea is to show that the limit on the right-hand side of this
inequality is $\frac{J(\theta)}{8}$ (same as the left-hand side),
and hence, the limit of the middle term is also
$\frac{J(\theta)}{8}$.
\begin{eqnarray}
&& \hspace*{-0.7cm} \lim_{\theta' \rightarrow \theta}
\frac{\varepsilon_1^2}{2 \sigma_1^2 (\theta-\theta')^2}
\nonumber \\[0.1cm]
&& \hspace*{-0.7cm} \stackrel{(\text{a})}{=} \lim_{\theta'
\rightarrow \theta} \frac{D(P_\theta ||
P_{\theta'})^2}{2 \sigma_1^2 (\theta-\theta')^2} \nonumber \\[0.1cm]
&& \hspace*{-0.7cm} \stackrel{(\text{b})}{=} \frac{J(\theta)}{4}
\lim_{\theta' \rightarrow \theta} \frac{D(P_\theta ||
P_{\theta'})}{\sigma_1^2} \nonumber \\[0.1cm]
&& \hspace*{-0.7cm} \stackrel{(\text{c})}{=} \frac{J(\theta)}{4}
\lim_{\theta' \rightarrow \theta} \frac{D(P_\theta ||
P_{\theta'})}{\sum_{x \in \mathcal{X}} P_\theta(x) \left( \ln
\frac{P_\theta(x)}{P_{\theta'}(x)} - D(P_\theta ||
P_{\theta'}) \right)^2} \nonumber \\[0.1cm]
&& \hspace*{-0.7cm} \stackrel{(\text{d})}{=} \frac{J(\theta)}{4}
\lim_{\theta' \rightarrow \theta} \frac{D(P_\theta ||
P_{\theta'})}{\sum_{x \in \mathcal{X}} P_\theta(x) \left( \ln
\frac{P_\theta(x)}{P_{\theta'}(x)} \right)^2 \; - \; D(P_\theta ||
P_{\theta'})^2} \nonumber \\[0.1cm]
&& \hspace*{-0.7cm} \stackrel{(\text{e})}{=} \frac{J(\theta)^2}{8}
\lim_{\theta' \rightarrow \theta}
\frac{(\theta-\theta')^2}{\sum_{x \in \mathcal{X}} P_\theta(x)
\left( \ln \frac{P_\theta(x)}{P_{\theta'}(x)}
\right)^2 \; - \; D(P_\theta || P_{\theta'})^2} \nonumber \\[0.1cm]
&& \hspace*{-0.7cm} \stackrel{(\text{f})}{=} \frac{J(\theta)^2}{8}
\lim_{\theta' \rightarrow \theta}
\frac{(\theta-\theta')^2}{\sum_{x \in \mathcal{X}} P_\theta(x)
\left( \ln \frac{P_\theta(x)}{P_{\theta'}(x)}
\right)^2} \nonumber \\[0.1cm]
&& \hspace*{-0.7cm} \stackrel{(\text{g})}{=} \frac{J(\theta)}{8}
\label{eq: first calculated limit}
\end{eqnarray}
where equality~(a) follows from \eqref{eq: epsilon1,2},
equalities~(b), (e) and~(f) follow from \eqref{eq: relation
between the divergence and Fisher information},
equality~(c) follows from \eqref{eq: sigma1 squared for the jumps
of the martingale U} with $P_1 = P_{\theta}$ and $P_2 =
P_{\theta'}$, equality~(d) follows from the definition of the
divergence, and equality~(g) follows by calculus (the required
limit is calculated by using L'H\^{o}pital's rule twice) and from
the definition of Fisher information in \eqref{eq: Fisher
information}. Similarly, also
$$ \lim_{\theta' \rightarrow \theta} \frac{\varepsilon_2^2}{2
\sigma_2^2 (\theta-\theta')^2} = \frac{J(\theta)}{8}$$ so
$$ \lim_{\theta' \rightarrow \theta} \, \min_{i=1,2}
\left\{ \frac{\varepsilon_i^2}{2 \sigma_i^2 (\theta-\theta')^2}
\right\} = \frac{J(\theta)}{8}.$$ Hence, it follows from
\eqref{eq: bounds on the limit} that $ \lim_{\theta' \rightarrow
\theta} \frac{E_{\text{L}}(P_\theta, P_{\theta'})}{(\theta -
\theta')^2} = \frac{J(\theta)}{8}.$ This completes the proof of
\eqref{eq: relation between the improved lower bound on the error
exponent and Fisher information}.

We prove now Eq.~\eqref{eq: relation between the loosened lower
bound on the error exponent and Fisher information}. From
\eqref{eq: d1}, \eqref{eq: d2}, \eqref{eq: delta1 and delta2} and
\eqref{eq: tilde E_L}
\begin{equation*}
\widetilde{E}_{\text{L}}(P_\theta, P_{\theta'}) = \min_{i=1,2}
\frac{\varepsilon_i^2}{2 d_i^2}
\end{equation*}
with $\varepsilon_1$ and $\varepsilon_2$ in \eqref{eq:
epsilon1,2}. Hence,
\begin{equation*}
\lim_{\theta' \rightarrow \theta}
\frac{\widetilde{E}_{\text{L}}(P_\theta, P_{\theta'})}{(\theta' -
\theta)^2} \leq \lim_{\theta' \rightarrow \theta}
\frac{\varepsilon_1^2}{2 d_1^2 (\theta' - \theta)^2}
\end{equation*}
and from \eqref{eq: first calculated limit} and last inequality
then it follows that
\begin{eqnarray}
&& \hspace*{-1.2cm} \lim_{\theta' \rightarrow \theta}
\frac{\widetilde{E}_{\text{L}}(P_\theta, P_{\theta'})}{(\theta' - \theta)^2} \nonumber \\
&& \hspace*{-1.2cm} \leq \frac{J(\theta)}{8}
\lim_{\theta' \rightarrow \theta} \frac{\sigma_1^2}{d_1^2} \nonumber \\
&& \hspace*{-1.2cm} \stackrel{(\text{a})}{=} \frac{J(\theta)}{8}
\lim_{\theta' \rightarrow \theta} \frac{\sum_{x \in \mathcal{X}}
P_\theta(x) \left( \ln \frac{P_\theta(x)}{P_{\theta'}(x)} -
D(P_\theta || P_{\theta'}) \right)^2}{\biggl( \max_{x \in
\mathcal{X}} \left|  \ln \frac{P_\theta(x)}{P_{\theta'}(x)} -
D(P_\theta || P_{\theta'}) \right| \biggr)^2}. \label{eq: second
calculated limit}
\end{eqnarray}
It is clear that the second term on the right-hand side of
\eqref{eq: second calculated limit} is bounded between zero and
one (if the limit exists). This limit can be made arbitrarily
small, i.e., there exists an indexed family of probability mass
functions $\{P_{\theta}\}_{\theta \in \Theta}$ for which the
second term on the right-hand side of \eqref{eq: second calculated
limit} can be made arbitrarily close to zero. For a concrete
example, let $\alpha \in (0,1)$ be fixed, and $\theta \in
\reals^+$ be a parameter that defines the following indexed family
of probability mass functions over the ternary alphabet
$\mathcal{X} = \{0, 1, 2\}$: $$ P_{\theta}(0) =
\frac{\theta(1-\alpha)}{1+\theta}, \quad P_{\theta}(1) = \alpha,
\quad P_{\theta}(2) = \frac{1-\alpha}{1+\theta}. $$ Then, it
follows by calculus that for this indexed family
\begin{equation*}
\lim_{\theta' \rightarrow \theta} \frac{\sum_{x \in \mathcal{X}}
P_\theta(x) \left( \ln \frac{P_\theta(x)}{P_{\theta'}(x)} -
D(P_\theta || P_{\theta'}) \right)^2}{\biggl( \max_{x \in
\mathcal{X}} \left|  \ln \frac{P_\theta(x)}{P_{\theta'}(x)} -
D(P_\theta || P_{\theta'}) \right| \biggr)^2} = (1-\alpha) \theta
\end{equation*}
so, for any $\theta \in \reals^{+}$, the above limit can be made
arbitrarily close to zero by choosing $\alpha$ close enough to~1.
This completes the proof of \eqref{eq: relation between the
loosened lower bound on the error exponent and Fisher
information}, and also the proof of Proposition~\ref{proposition:
Fisher information}.


\begin{thebibliography}{99}
\bibitem{AlonS_tpm3}
N. Alon and J. H. Spencer, {\em The Probabilistic Method}, Wiley
Series in Discrete Mathematics and Optimization, Third Edition,
2008.
\bibitem{Azuma}
K. Azuma, ``Weighted sums of certain dependent random variables,''
{\em Tohoku Mathematical Journal}, vol.~19, pp.~357--367, 1967.
\bibitem{Billingsley}
P. Billingsley, {\em Probability and Measure}, Wiley Series in
Probability and Mathematical Statistics, Third Edition, 1995.
\bibitem{Blahut_IT74}
R. E. Blahut, ``Hypothesis testing and information theory,'' {\em
IEEE Trans. on Information Theory}, vol.~20, no.~4, pp.~405--417,
July 1974.
\bibitem{survey2006}
F. Chung and L. Lu, ``Concentration inequalities and martingale
inequalities: a survey,'' {\em Internet Mathematics}, vol.~3,
no.~1, pp.~79--127, March 2006. [Online]. Available:
\url{http://www.ucsd.edu/~fan/wp/concen.pdf}.
\bibitem{Cover and Thomas}
T. M. Cover and J. A. Thomas, {\em Elements of Information
Theory}, John Wiley and Sons, second edition, 2006.
\bibitem{Csiszar_Shields_FnT}
I. Csisz\'{a}r and P. C. Shields, {\em Information Theory and
Statistics: A Tutorial}, Foundations and Trends in Communications
and Information Theory, vol.~1, no.~4, pp.~417--528, 2004.
\bibitem{Dembo_Zeitouni}
A. Dembo and O. Zeitouni, {\em Large Devitations Techniques and
Applications}, Springer, second edition, 1997.
\bibitem{Hoeffding}
W. Hoeffding, ``Probability inequalities for sums of bounded
random variables,'' {\em Journal of the American Statistical
Association}, vol.~58, no.~301, pp.~13--30, March 1963.
\bibitem{Hollander_book_2000}
F. den Hollander, {\em Large Deviations}, Fields Institute
Monographs, American Mathematical Society, 2000.
\bibitem{McDiarmid_tutorial}
C. McDiarmid, ``Concentration,'' {\em Probabilistic Methods for
Algorithmic Discrete Mathematics}, pp.~195--248, Springer, 1998.
\bibitem{Sason_submitted_paper}
I. Sason, ``On refined versions of the Azuma-Hoeffding inequality
with applications in information theory,'' last updated in July 2012.
[Online]. Available: \url{http://arxiv.org/pdf/1111.1977v5.pdf}.
\end{thebibliography}
\end{document}